\documentclass{svproc}
\makeatletter
\newcommand{\@chapapp}{\relax}%
\makeatother
\usepackage[title,toc,titletoc,header]{appendix}

\usepackage{url}

\usepackage{graphicx}%
\usepackage{multirow}%
\usepackage{amsmath,amssymb,amsfonts}%
\usepackage{mathrsfs}%
\usepackage{xcolor}%
\usepackage{textcomp}%
\usepackage{manyfoot}%
\usepackage{booktabs}%
\usepackage{algorithm}%
\usepackage{algorithmicx}%
\usepackage{algpseudocode}%
\usepackage{listings}%

\usepackage{xcolor, colortbl}
\definecolor{LightCyan}{rgb}{0.88,1,1}
\definecolor{Gray}{gray}{0.85}
\usepackage{tikz}

\newcommand{\circled}[1]{\tikz[baseline=(char.base)]{
    \node[shape=circle,draw,inner sep=1pt] (char) {#1};}}
    
\newcommand*\emptycirc[1][0.8ex]{\tikz\draw (0,0) circle (#1);} 

\newcommand*\fullcirc[1][0.8ex]{\tikz\fill (0,0) circle (#1);} 

\usepackage{tabularx}
\usepackage{makecell}
\usepackage[flushleft]{threeparttable}

\begin{document}
\mainmatter              
\title{ExPrESSO: Zero-Knowledge backed Extensive Privacy Preserving Single Sign-on}
\titlerunning{ExPrESSO: Extensive Privacy Preserving Single Sign-on}  
%
\author{Kaustabh Barman \and Fabian Piper \and
Sanjeet Raj Pandey \and Axel Küpper}
\authorrunning{Kaustabh Barman et al.} 
%
\tocauthor{First Author, Second Author, Third Author and Fourth Author}
\institute{Technische Universität Berlin, Berlin, Germany\\}


\maketitle              

\begin{abstract}
User authentication is one of the most important aspects for secure communication between services and end-users over the Internet. Service providers leverage Single-Sign On (SSO) to make it easier for their users to authenticate themselves. However, standardized systems for SSO, such as OIDC, do not guarantee user privacy as identity providers can track user activities. We propose a zero-knowledge-based mechanism that integrates with OIDC to let users authenticate through SSO without revealing information about the service provider. Our system leverages Groth's zk-SNARK to prove membership of subscribed service providers without revealing their identity. We adopt a decentralized and verifiable approach to set up the prerequisites of our construction that further secures and establishes trust in the system. We set up high security targets and achieve them with minimal storage and latency cost, proving that our research can be adopted for production.
\keywords{Zero-Knowledge Proofs, ZoKrates, zk-SNARK, OpenID Connect, Privacy, Multi-Party Computation}
\end{abstract}
\section{Introduction}\label{sec1}

In the early web, internet-based service providers (SP) maintained a siloed identity model, requiring users to manage separate accounts and passwords for every SP. As the internet grew, this led to “password fatigue” with users juggling dozens of accounts~\cite{florencio2007large}. In response, the 2000s and 2010s saw the rise of single sign-on (SSO) and federated identity solutions that allow one account to authenticate across multiple services. SSOs make it easier for the users to remember and manage single login-password credentials for multiple SPs while also reducing the security risk of weak passwords or password reusage~\cite{pfitzmann2003analysis}. With SSO systems, SPs also don't need to manage and secure huge password databases to authenticate their consumers, unlike siloed identities~\cite{yildiz2021connecting}. However, the convenience of SSO's easy login can obscure the privacy costs, as IdPs are able to easily monitor user logins and activities across various SPs~\cite{guo2022up}. 

SSO protocols such as OpenID Connect (OIDC)~\cite{sakimura2014openid} and Security Assertion Markup Language (SAML)~\cite{hughes2005security} are often deployed by Identity Providers (IdPs) for SPs to authenticate end-users. OIDC is an identity layer specification from OpenID Foundation (OIDF) designed on top of the OAuth 2.0 protocol specification~\cite{hardt2012rfc} published by IETF~\cite{10434479}. The OIDC specification was designed to verify the identity of end-users based on the identity credentials issued by the IdP. It includes obtaining profile information about the end-user in an interoperable and REST-like manner. Authentication using identity credentials has been a lively research topic, from both a security and privacy point of view. Most privacy-oriented work proposes protocols to replace OIDC, often requiring prerequisite technical knowledge or using external tools by the end-user. In our work, we discuss the privacy issues of OIDC and propose a mechanism to enhance user privacy without involving any additional user-bound action.  In OIDC, the most prominently used variants are the implicit flow and the authorization code flow. Since only the implicit flow does not require direct backend communication between the IdP and the client, only the implicit flow has any room for privacy-preserving techniques for the end-user~\cite{kroschewski2023save}.

\textbf{Contribution.} The current OIDC flows do not allow the end-user any kind of privacy from the IdP. The IdP can easily record and monitor every authentication activity of every end-user with every RP. Consequently, it becomes very easy for IdPs to maintain a behavioral profile of every end-user login activity. In our research, we propose Extensive Privacy Preserving Single Sign-On (ExPrESSO), an OIDC-based framework for anonymization of service providers from IdPs. We define certain security aspects that ExPrESSO achieves that enable user privacy both from IdPs and colluding RPs. ExPrESSO exploits properties of Zero-Knowledge Proof (ZKP) schemes to address user-privacy issues with IdPs in OIDC. We analyze current shortcomings in OIDC and build a Multi-Party Computation (MPC) based mechanism to establish trust and enable secure exchange of proofs in a trustless environment. ExPrESSO assumes that end-users do not posses any technical knowledge and does not want to divert to usage of new tools for privacy. To the best of our knowledge, no other work has applied such an approach to enhance privacy in SSOs. We further compare our results with other works that have tried to achieve similar goals as ours and provide a detailed evaluation of our comparison.



\section{Preliminaries}\label{background}
This section introduces the background necessary for understanding our work. First, the OIDC protocol, its variants, and relevant concepts are introduced. Next, we discuss ZKPs, their properties, and how they have developed over time to be more efficient and easier to use. We will also look into the various flavors of ZKPs, which one fits best for our use case, and the underlying primitives.

\subsection{OpenID Connect}
OIDC is an identity authentication protocol that is built on top of the industry-standard authorization protocol OAuth 2.0. In OIDC, the IdP deploys the authorization server that can authenticate end-users. SPs become Relying Parties (RP)/clients that send authorization requests to IdP and receive user identity attributes in response. Hence, the terms RP and clients can be used interchangeably. SPs encapsulate the identity attributes in a JSON Web Token and send it to RP only after successfully authenticating a user. Our work builds upon a specific variant of OIDC flow and also includes pseudonymous pairing for user privacy from colluding SPs. 

\subsubsection{OIDC Flows}
Since OIDC is realized on top of OAuth 2.0, it also specifies three prominent flows similar to OAuth 2.0~\cite{sakimura2014openid}: 

\begin{enumerate}
    \item \textbf{Implicit Flow: }The implicit flow is for lightweight RPs without any backend, such as Single Page Applications (SPA). Relying parties authenticate their users by simply redirecting the user to the login page of the IdP. The user logs in and obtains a \textit{id\_token} enclosing the user's identity attributes. The \textit{id\_token} is forwarded to the RP through the user agent, typically a browser. The RP can then use the \textit{id\_token} to extract the user's attributes and authenticate the user. 
    \item \textbf{Code Flow: }When the RP is capable of executing backend operations, it can follow the more secure code flow. Here, the end-user obtains an authentication \textit{code} instead of a \textit{id\_token}, which is forwarded to the RP. The \textit{code} is a short-lived one-time usage string bound to the user's consent to its identity attributes. The RP reaches out to the IdP from its backend and exchanges the \textit{code} for a \textit{id\_token} containing the user-permitted identity attributes. Since this step is necessary to obtain the \textit{id\_token}, the IdP can track the RP's IP address, which means the IdP is always able to identify the RP. Hence, the room for privacy is really less in the code flow. On obtaining the \textit{id\_token}, the RP can extract the user's attributes and authenticate the consumer, similar to the implicit flow.  
    \item \textbf{Hybrid Flow: }The hybrid flow combines both implicit and code flow. It offers immediate issuance of \textit{id\_token} to clients through implicit flow. Additionally, it also adds the exchange authentication code for further identity attributes. This flow can be useful for use cases where the client requires user authentication before running backend logic.  
\end{enumerate}
 IdPs usually maintain a list of all RPs that register with them to use their SSO service. An IdP identifies an RP during user authentication using a unique RP-identifier called \textit{client\_id}. The \textit{client\_id} is generated by the IdP and is returned to the RP right after a SP has registered as a client with the IdP's identity service, which takes place during a pre-registration phase beforehand. 
 
 In our work, we aim to prevent the IdP from finding out which RP an end-user wants to access during authentication. It means that during an SSO session, the RP must not disclose its identifier to the IdP. Furthermore, any direct HTTP communication between the IdP and the RP might also reveal the RP identity using IP address tracking. Since only the implicit flow does not require any direct interaction between the RP and the IdP, our work builds on the assumption that the actors are only using the implicit flow for SSO sessions.

\subsubsection{Pairwise Pseudonymous Identifiers (PPID)}
NIST outlined that PPIDs allow IdPs to provide multiple distinct federated identifiers to different RPs for a single end-user's account~\cite{nist-sp800-63c}. As the name suggests, PPIDs are pseudonymous identifiers linked with pairs of user-client relations. PPIDs can prevent multiple colluding RPs from linking user identity attributes together to access attributes of user identities without their consent. OIDC typically recommends IdPs to map each user identifier to the associated client identifiers that request the user identity attributes. Each combination of a user identifier and the client identifier is allocated a PPID that always remains the same for all subsequent user authentication sessions. Hence, a client always gets the same PPID for a given user, but another client will receive a different PPID for the same user. 

\subsection{Zero-Knowledge Proof}
 Zero-knowledge proofs were first introduced by Goldwasser et al. to enable a prover to convince that a statement was true without revealing any other information about the statement~\cite{10.1145/3335741.3335750}. ZKPs have three core properties:
\begin{itemize}
    \item \textbf{Completeness: }With a given statement and a witness, a prover can convince a verifier about the correctness of a statement.
    \item \textbf{Soundness: }A malicious prover is not able to convince the verifier.
    \item \textbf{Zero-knowledge: }The proof does not reveal anything but the truth of the statement, specifically not revealing the witness. 
\end{itemize}

ZKPs are either interactive (multiple message exchanges) or non-interactive (NIZK), where a single proof suffices. To convince a verifier about the correctness of a statement, the verifier must carry out a cryptographic computation to verify a proof. If the proof is computationally expensive to verify, verifiers with limited computational resources will struggle to verify it. Hence, the term Succinct Non-Interactive Argument of Knowledge (SNARK) was coined with efficient proofs for the verifier, but still computational overhead for the prover was orders of magnitude too high~\cite{bitansky2012extractable}\cite{walfish2015verifying}. Jens Groth in 2016 proposed the zk-SNARK primitive scheme, also known as Groth16, that enables small and easy-to-verify proof~\cite{groth2016size}. Since the inception of zk-SNARKs, cryptography experts and researchers have used it to enhance privacy in decentralized applications and blockchain transactions of cryptoassets such as Zerocash~\cite{6956581}.

A ZKP setup consists of a prover that proves the knowledge of a statement with a proof. A verifier verifies the proof without the statement being revealed. The context of the statement can vary depending on the use case. For a zk-SNARK, a simple function can be constructed with a conditional statement to assert the actual statement. The function may simply return true if the statement satisfies the condition, and return false otherwise. Such a function can be converted into an arithmetic circuit with Rank 1 Constraint System (R1CS). In our work, we use Groth16~\cite{groth2016size} scheme to implement such a mechanism between a prover and a verifier. Groth16 generates a Common Reference String (CRS) for a given arithmetic circuit $\mathcal{C} : \mathbb{F}_p^{n_w + n_x} \to \mathbb{F}_p^{n_o}$ where $n_w$ is the number of private witness variables, $n_x$ is the number of public inputs, $n_o$ is the number of outputs in the domain over the prime field $\mathbb{F}_p$. A valid proof shows $C(x,w)=1$ without revealing $w$. Groth16 needs a Baby Jubjub elliptic curve encoding of polynomials related to the circuit. Baby Jubjub~\cite{whitehat2020baby} belongs to the sub-class of twisted Edwards curve optimized for circuit efficiency in zk-SNARKS.
Below, we elaborate on the different stages of the mechanism. 

\subsubsection{Trusted Setup Ceremony}
The setup generates a Common Reference String (CRS) from random trapdoors $\tau, \alpha, \beta, \gamma, \delta \in F_p$. If the randomness used in this phase is known, proofs can be forged. Hence, it proceeds in two phases:
\begin{enumerate}
    \item Phase 1: \textit{Powers of Tau}  \newline
    In this phase, circuit-independent parameters $[\tau^i \cdot G]_1$ for all degrees of $i$ are generated once and reused. The Ethereum community performs a \textit{Powers of Tau} ceremony and maintains the outputs in the repository called \textit{Perpetual Powers of Tau}~\cite{powersoftau}. For the sake of simplicity, we can use the parameters from \textit{Perpetual Powers of Tau}. We select the smallest parameter with a number of points greater than the number of constraints in $\mathcal{C}$.
    \item Phase 2: \textit{Multi-Party Trapdoor Computation (MPC)} \newline
    In the second phase, trapdoor-based encoding without any single party knowing the trapdoor is computed. To keep the secret trapdoors unknown, we adopt MPC based trapdoor computation by Ben-Sasson et. al.~\cite{mpc} that includes multiple participants in the ceremony. Each participant contributes shares $\alpha_i, \beta_i, \gamma_i, \delta_i$, combined as:
    \[\alpha = \prod_i \alpha_i, \quad \beta = \prod_i \beta_i, \quad \gamma = \prod_i \gamma_i, \quad \delta = \prod_i \delta_i\]
    Participants compute and contribute group elements and apply the trapdoors to powers of $\tau$:
    \[\alpha \cdot \tau^i \cdot G, \quad \beta \cdot \tau^i \cdot G, \quad \tau^i/\delta \cdot G, \quad \text{for all required } i\]
    As long as one party deletes its share, trapdoors remain unknown. These computations are also verifiable and can be made publicly available for verification to establish trust in the setup ceremony~\cite{mpc}. The final composition of CRS takes place by evaluating a proving key $PK$ and a verification key $VK$ as:
\[PK = \{\tau \cdot G, \; \alpha \cdot G, \; \beta \cdot G, \; \gamma \cdot G, \; \delta \cdot G, \; \text{and constraint polynomials at } \tau\}\]
where the constraint polynomials at $\tau$ include $\alpha \cdot \tau^i \cdot G$, $\beta \cdot \tau^i \cdot G$, and $\tau^i/\delta \cdot G$. The $VK$ is computed as:
\[VK = \{\alpha \cdot G_1, \; \beta \cdot G_2, \; \gamma \cdot G_2, \; \delta \cdot G_2, \; \mathbf{Y}_1 \cdot G_1, \; \mathbf{Y}_2 \cdot G_1, \; \mathbf{Y}_3 \cdot G_1\}\]
where each $\mathbf{Y}_i$ is the linear combination of public inputs evaluated at $\tau$. \newline
$G_1$ and $G_2$ are generators of the respective elliptic curve groups.
\end{enumerate}

\subsubsection{Proof Generation}
Using $PK$, the prover computes the zk-SNARK proof $\pi=(A,B,C)$ for $\mathcal{C}$ by computing:
\[A = \alpha \cdot G_1 + \sum_{i} w \mathbf{A}_i(\tau) \cdot G_1\]
\[B = \beta \cdot G_2 + \sum_{i} w \mathbf{B}_i(\tau) \cdot G_2\]
\[C = \gamma \cdot G_1 + \sum_{i} w \mathbf{C}_i(\tau) \cdot G_1\] 
The proof is generated such that it can only be verified if the private inputs satisfy the circuit constraints. Such a proof is small in size, easy and fast to verify by a verifier.

\subsubsection{Verification}
The verifier checks bilinear pairings using $VK$:
\[
e(A, B) = e(C, G) \cdot e(\mathbf{x}_1 \cdot G_1, \beta \cdot G_2) \cdot e(\mathbf{x}_2 \cdot G_1, \gamma \cdot G_2) \cdot e(\mathbf{x}_3 \cdot G_1, \delta \cdot G_2)
\]
If the equation holds, the proof is valid and the witness remains hidden.

\section{Design Goals}\label{design goals}
Our work begins by realizing potential adversaries and their behavior to formulate our adversary model. Furthermore, we define the characteristics of our systems that will help mitigate adversary attacks. 

\subsection{Adversary Model}
Our model works with four actors in order to deploy a functioning privacy-preserving SSO system: \newline
\begin{enumerate}
 \item \textbf{OIDF:} 
    As OIDF is a non-profit open standards body developing identity and security specifications. The OIDF is led by several reputed legal entities (industry partners) with expert representatives from the identity industry~\cite{oidf}. OIDF-side operations and decisions are always the result of a distributed and collaborative effort. Hence, we can safely assume OIDF to be an honest and neutral actor. In our system, OIDF acts as the trust anchor that publishes standardized and compilable programs and also provides RESTful services for integrity checks for cryptographic artifacts. \newline
\item \textbf{IdP:} 
    The IdP is the central identity management entity that safeguards user's identity attributes. Based on the degree of the user's consent, it can issue identity credentials and access tokens to SPs. The IdP is assumed to be an honest but curious entity, meaning it follows the protocol honestly but might attempt to infer additional information about user authentication activities. It might actively monitor login activities to infer user behavior. It might attempt to track which clients the users are logging into and build behavioral profiles over time. \newline
\item \textbf{User:} 
    End-users are interested in authenticating to RPs using their identity managed by the IdP. Users interact with both the RP and IdP during the authentication process, typically using a user-agent such as a web browser or mobile application that facilitates interactions between the user, IdP, and RP. In our adversary assumptions, a dishonest user may attempt to authenticate to an SP that is not registered (through the registration phase) or de-registered from the IdP. \newline
\item \textbf{RP:} 
    IdPs maintain a list of clients that are SPs that are registered with the IdP's identity service. The IdP must prevent issuing authentication tokens to unregistered or deregistered RPs while maintaining privacy guarantees for legitimate logins. Dishonest colluding RPs might also try to access more identity attributes of a user than the user has originally consented to sharing. \newline
\end{enumerate}

In our work, we model the framework in such a way that any of the actors behaving dishonestly will not be able to successfully breach the privacy of the end-user. Hence, to counteract such adversarial threats, the system aims to achieve:
\begin{enumerate}
\item \textbf{RP Unobservability:} 
The system aims to achieve RP unobservability, ensuring that the IdP cannot determine which RP a user is attempting to authenticate with. 
Traditional OIDC authentication exposes the RP’s identity to the IdP in each authentication request when the RP mentions its client identifier. 
The system should prevent a curious IdP from learning the RP’s identity while still enabling secure RP validation.

\item \textbf{User Unlinkability:} 
Ensuring user unlinkability prevents multiple RPs from colluding to infer that two authentication events belong to the same user.
OIDC recommends the use of its Pairwise Pseudonymous Identifier feature \cite{sakimura2014openid} to hide the identity of a user from an RP.
This requires knowledge of the identity of RP for the IdP, which breaks with our design goal of RP unobservability.
Hence, our system needs to ensure that a user receives a unique, unlinkable pseudonym similar to the OIDC Pairwise Pseudonymous Identifiers, without exposing the identity of RP to IdP.

\item \textbf{Authentication Integrity:} 
Authentication integrity guarantees that only registered RPs can obtain valid authentication tokens from the IdP. Without this restriction, a malicious entity could attempt to impersonate an RP and receive a valid token, leading to unauthorized access or token misuse. 
The system must ensure that authentication requests originate only from pre-registered and authenticated RPs.  

\item \textbf{Client Revocation:} 
The system should also let IdPs revoke RP's membership when an RP deregisters from the IdP's identity service. The revoked RP should not be able to obtain a user identity token, even if it possesses old data/credentials from the IdP. 

\item \textbf{Trust Distribution:}
Trust distribution aims to reduce the need for implicit trust in a single party.
Verifiable mechanisms and trust minimization should be employed to mitigate single responsibilities.
The approach should take a holistic perspective, focusing not only on the login process but also on registering and deregistering an RP.

\end{enumerate}
Such security aspects will mitigate adversaries including internal dishonest actors of the system and external attackers.  


\subsection{Expected Characteristics}
We want to design ExPrESSO to integrate seamlessly with existing protocols such as OIDC without requiring any additional effort from end-users. Users should not need to install applications, plugins, or tools; once IdP and RP agree to use ExPrESSO, privacy-preserving authentication is achieved transparently. The framework assumes a trustless environment, where no single entity is fully relied upon, and any necessary trust anchors (e.g., standardization or setup) are distributed among multiple parties. Finally, the system is intended to be practical, adding minimal computational and communication overhead so that it can be deployed in real-world scenarios.


\section{ExPrESSO Setup and Architecture}\label{sec8}
Our system comprises a sequence of phases involving the actors in OIDC-based authentication. Following IETF’s OAuth 2.0 Dynamic Client Registration Protocol~\cite{richer2015oauth} and OpenID’s Dynamic Client Registration~\cite{sakimura2014clientreg}, we include an RP registration phase that occurs before user authentication. Only registered RPs can later authenticate users through the IdP’s service. We describe how ExPrESSO conceals RP identity over the issued \textit{client\_id}, and how this credential is used during authentication. We also cover membership revocation when RPs deregister from the IdP. Each phase is independent but must follow a sequential order, i.e., a later phase can only occur after its prerequisites. Figure~\ref{fig:expresso architecture} illustrates this architecture, showing actor interactions across registration and authentication with an MPC-based trusted setup.


\begin{figure*}[h]
  \centering
  \includegraphics[width=1.0\textwidth]{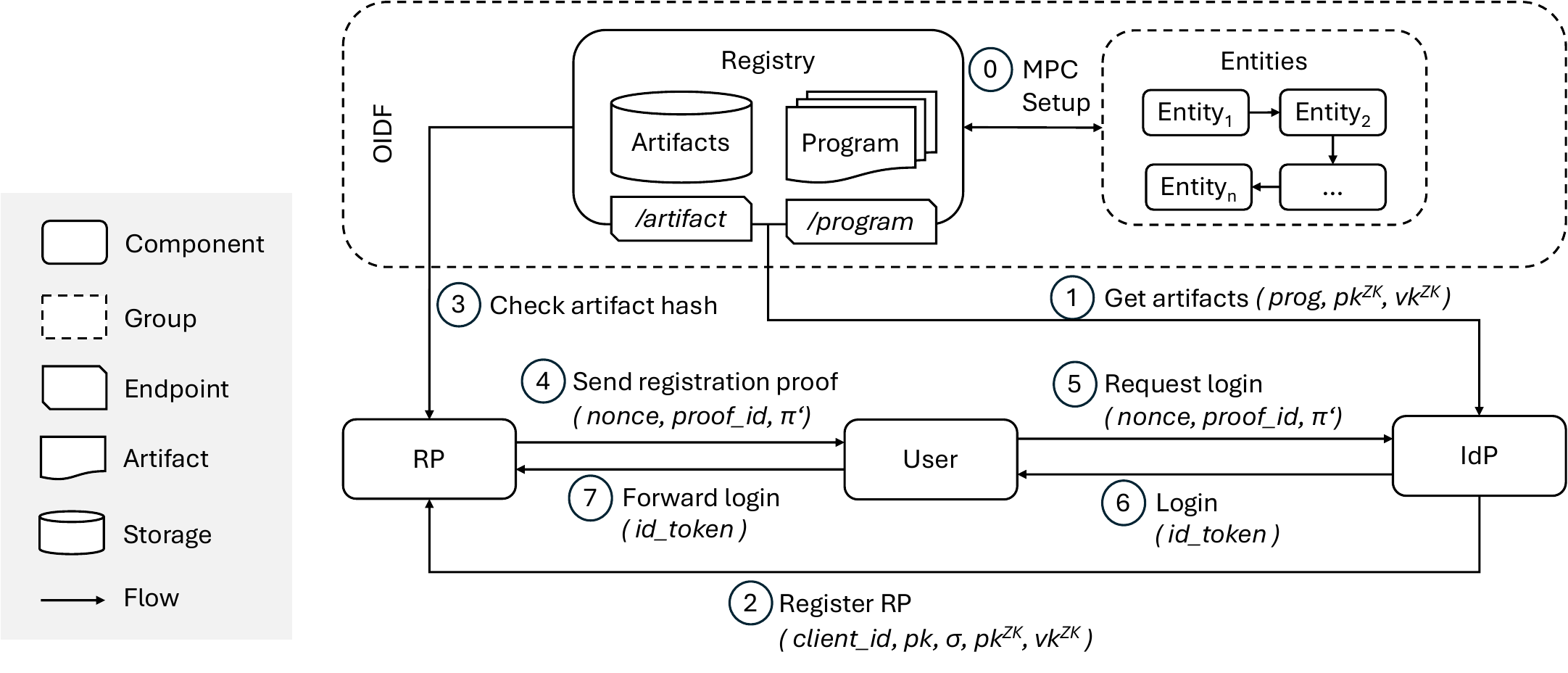}
  \caption{Overview of ExPreSSO Architecture with MPC-based trusted setup}
  \label{fig:expresso architecture}
\end{figure*}

\vspace{-10px}
\subsection{Client Registration}
The RP (client) registers with the IdP to allow its users to authenticate at a later point in time using the identity service of the IdP. During user authentication, the IdP checks if the user is authenticating to an RP that has registered and subscribed to the IdP's identity service. In OIDC, the IdP issues a \textit{client\_id} to the RP along with an optional \textit{client\_secret} and all other registered metadata about the RP as a response~\cite{sakimura2014clientreg}. In our system, the IdP additionally generates a public key and signs the \textit{client\_id} before returning them as a set of credentials (public key, \textit{client\_id}, and signature). The IdP later expects the RP to provide a zk-SNARK-based proof of membership for such \textit{client\_id} during the user authentication session without revealing the \textit{client\_id}. As explained in section~\ref{background}, this requires a trusted ceremony to be performed. Hence, we assume OIDF to set up a secure environment such that the RP can create and exchange a proof in a trustworthy manner. In Figure~\ref{fig:expresso architecture}, we describe the operations to be performed by OIDF, IdP, and RP, which are represented by step \circled{0} to \circled{2}.

\subsubsection{Circuit Design by OIDF.}
 OIDF is managed by a group of entities to anchor trust in a decentralized manner. The entities collaborate to create and agree to a boilerplate program, a high-level, human-readable, and compilable set of instructions. The boilerplate should clearly mention public and private inputs to a function that verifies the authenticity of a given signature for a message against a public key. An example of such boilerplate code is shown in algorithm~\ref{alg:verifySign}. In our model, we assume that the code is standardized, open-sourced, and published by the OIDF. This eliminates any possibility for curious IdPs to hide the scope of the parameters from the RP and establish trust in the circuit. Hence, both IdP and RP can trust it as the scopes of the inputs are transparent. Furthermore, the code is used as the basis for the required arithmetic circuit with R1CS in all of the stages of Groth16's primitive.

\subsubsection{OIDF's MPC-Based Trusted Ceremony.}\label{oidf-mpc}
Other than standardizing the boilerplate program, the OIDF plays another important role in our system by performing the trusted ceremony. As described in section~\ref{background}, a trusted ceremony is necessary to generate a set of proving and verification keys. In our work, OIDF uses \textit{Pertual Powers of Tau} for the first phase of the ceremony. For the second phase, OIDF adopts Ben-Sasson et. al's MPC-based trapdoor computation~\cite{mpc}. Each entity (partners) in OIDF must contribute their version of randomized secret trapdoors and public parameters. 
OIDF can run multiple cycles of the 2-phase trusted-setup ceremony for the same R1CS circuit (deduced from the boilerplate program) to generate multiple sets of artifacts containing the proving and verification keys. We call the set of artifacts containing the boilerplate program, a proving key $PK$, and a verification key $VK$ as \textit{zk-Artifacts}. Each cycle generates a set of zk-Artifacts with different proving and verification keys for the same R1CS circuit. OIDF should maintain a pool of zk-Artifacts by running multiple cycles of the 2-phase ceremony. IdPs can request the zk-Artifacts from OIDF whenever necessary. Additionally, OIDF  publishes a hash of the issued zk-Artifacts \textit{H(zk-Artifacts)} publicly to prevent IdPs from acquiring multiple zk-Artifacts for tracking RPs. 

\begin{algorithm} 
\caption{High-level level program to verify signature $R,S$ with associated \textit{client\_id} and public key $pk$ of IdP}\label{alg:cap}
\begin{algorithmic}
\Function{verifySignature}{public $pk$, private $R$, private $S$, private $M$} $\rightarrow(bool)$
\State $h \gets hash(pk, M) $
\State $ lhs \gets S * G $
\State $ rhs \gets R+h*pk $ \newline
\Return $ lhs == rhs n$
\EndFunction
\end{algorithmic}
\label{alg:verifySign}
\end{algorithm}
\vspace{-15px}

\subsubsection{Exchange of RP Credentials and zk-Artifacts.}
In ExPrESSO, RP registration with the IdP's identity service is crucial. The IdP requests zk-Artifacts from OIDF, which are then prepared to be shared with the RP. During registration, the RP sends some typical RP information, such as client name and redirect URI. When the RP has registered with the IdP's identity service successfully, the IdP generates a corresponding unique \textit{client\_id} only valid for the respective RP. The IdP then signs the \textit{client\_id} and binds the \textit{client\_id}, its public key $pk$, and signature $R,S$ into an RP-credential. The RP-credentials and zk-Artifacts are sent to the RP as a response to successful registration. 

\subsection{User Authentication}
The user authentication takes place post RP-registration. In the typical OIDC flow, the user initiates the authentication process by selecting the IdP to use. The RP then redirects the user-agent to the IdP's login page. We assume that the RP uses a proxy redirection link to hide the real redirection URI from the IdP. The RP includes its \textit{client\_id}, requested scope, local state, and a redirection URI to which the IdP sends the user-agent back when access is granted~\cite{hardt2012rfc}. In ExPrESSO, similar to typical OIDC flow, the user also initiates the authentication process and gets redirected to the IdP's user authentication page. However, the RP does not send its \textit{client\_id} during the user authentication process in order to hide its identity for the user's privacy. After the user performs a conventional username-password-based authentication, the RP proves its membership with the IdP by sending a proof of membership using the deployed zk-SNARK scheme in ExPrESSO. Hence, the RP must check for artifact integrity by interacting with OIDF, then create a zk-SNARK-based proof of membership and send it to the IdP. The IdP authenticates the user and verifies the membership proof before returning an identity token. Each of these steps is explained in detail below and is represented in Figure~\ref{fig:expresso architecture} from step \circled{3} to \circled{7}.

\subsubsection{Interaction with OIDF.}
During user authentication, to prove RP's membership status with the IdP without revealing its real identifier, the RP must generate a zk-SNARK-based proof. However, before generating the proof, the RP should also verify the integrity of zk-Artifacts that it received from the IdP. This is to ensure that the zk-Artifacts obtained from the IdP are uniform with the most recent version of artifacts in use by all other clients. The RP verifies the integrity of zk-Artifacts by performing the following three steps:
\begin{enumerate}
    \item The RP obtains from OIDF the hash of the most recent version of zk-Artifact associated with the respective IdP.
    \item The RP hashes the zk-Artifact that it received from the IdP.
    \item The RP matches the two hashes to verify that the zk-Artifact from IdP is the most recent version. 
\end{enumerate}

\subsubsection{Membership Proof by RP.}
After making sure that the zk-Artifacts are uniform with the most recent version, the RP can proceed to create a zk-SNARK-based proof of membership. As explained in section~\ref{background}, to create such a proof, the RP (prover) executes the following steps:
\begin{enumerate}
    \item The RP compiles the published boilerplate program with pre-defined input parameters obtained from the OIDF's published endpoint. Compiling the boilerplate program results in the required R1CS circuit for proof generation.
    \item The RP then extracts the proving key from the zk-Artifact. 
    \item It computes a witness vector by providing all the public and private inputs to the circuit in the right order.
    \item The RP generates the proof $\pi$ with the three elements $\pi = \{A,B,C\}$.
\end{enumerate} 
After generating the proof, the RP redirects the user to the IdP along with the proof $\pi$.

\subsubsection{Proof Verification by IdP.}
When an IdP receives a user authentication request, it now contains the proof and the scope of the requested user's identity attributes in addition to other metadata. After the user authenticates itself and provides consent for the requested identity attributes, the IdP proceeds to verify the RP's membership proof. The membership proof can be verified to prove that the RP holds a signed \textit{client\_id} and an up-to-date proving key without revealing the \textit{client\_id}. The IdP performs the verification step as mentioned in section~\ref{background}. If the proof is verified positively, the IdP creates an identity token and signs it. 

\subsection{RP De-Registration}\label{revocation}
In ExPrESSO, it is also possible for RPs to deregister themselves from a given IdP's identity service. The IdP can revoke membership of RPs by an easy mechanism. When an RP deregisters itself, IdP can simply request a new set of zk-Artifacts from OIDF. IdP should then send (broadcast) the new set of artifacts to all remaining valid RP members and use only the new verification key from the new artifacts set to verify further membership proofs. De-registered RPs will not receive the corresponding new proving key, making them unable to generate membership proofs using the latest version of zk-Artifacts. 

\section{Implementation and Security Analysis}
In this section, we will discuss the systematic integration of zk-SNARK and PPIDs. We will further analyze how certain implementations of mechanisms achieve our security goals. 

\subsection{zk-SNARK Construction}
As explained in the previous section, all actors except the end-user must carry out cryptographic computations regarding the trusted setup ceremony, proving and verification. The verifiable computation schemes are implemented using the ZoKrates domain-specific language.
ZoKrates is a toolbox for zk-SNARK that can make verifiable computations from the specification of a high-level and compilable program to generate proofs of computation and verifying those proofs~\cite{eberhardt2018zokrates}.
It provides support for Groth16 and other schemes such as Marlin and GM17. 
We used ZoKrates to implement proof of membership in zero-knowledge. 

\textbf{IdP Bootstrapping.} The IdP must generate a public-private key pair to sign credentials during the client registration phase. We bootstrap the IdP by generating private key $sk \in \mathbb{Z}_q$ and Public Key $pk = sk.G$ where $G$ is the generator of the Baby Jubjub elliptic curve. 
The IdP is later able to use $sk$ to generate signatures for multiple \textit{client\_id} of clients. 
The IdP also compiles the boilerplate program to generate a R1CS-based arithmetic circuit. The R1CS output requires 400MB of disk space, including the circuit artifact and a corresponding circuit-bound Zokrates artifact.

\textbf{OIDF Functionalities.} The OIDF plays the role of providing the trust anchor by standardizing and open-sourcing the boilerplate program and publishing the hash value \textit{H(zk-Artifacts)} of IdPs. It also allows OIDF to maintain a pool of zk-Artifacts for multiple IdPs through a single boilerplate program instead of catering to different circuits from different IdPs, which will be slower in orders of magnitude for MPC-based artifact generation. Each zk-Artifact include a proving key of size 38.4MB and the verification key of size 4KB, making each set approximately 39MB in size. Depending on the number of artifacts OIDF decides to store in the pool at any time, the disk space requirement for OIDF will be a multiple of 39 MB and an additional 400MB for the R1CS output. This space requirement will always be constant since, after the allocation of an artifact to an IdP, it can be permanently deleted and replaced with a new artifact in the pool.

Next, as explained in section~\ref{oidf-mpc} the OIDF generates zk-Artifacts for all IdPs and maintains multiple versions of it. OIDF further makes the zk-Artifacts available to IdPs by exposing an endpoint. Additionally, whenever an IdP requests a zk-Artifact, the latest hash value of that artifact is published.

\begin{figure}[htbp]
  \centering
  \includegraphics[width=1\textwidth]{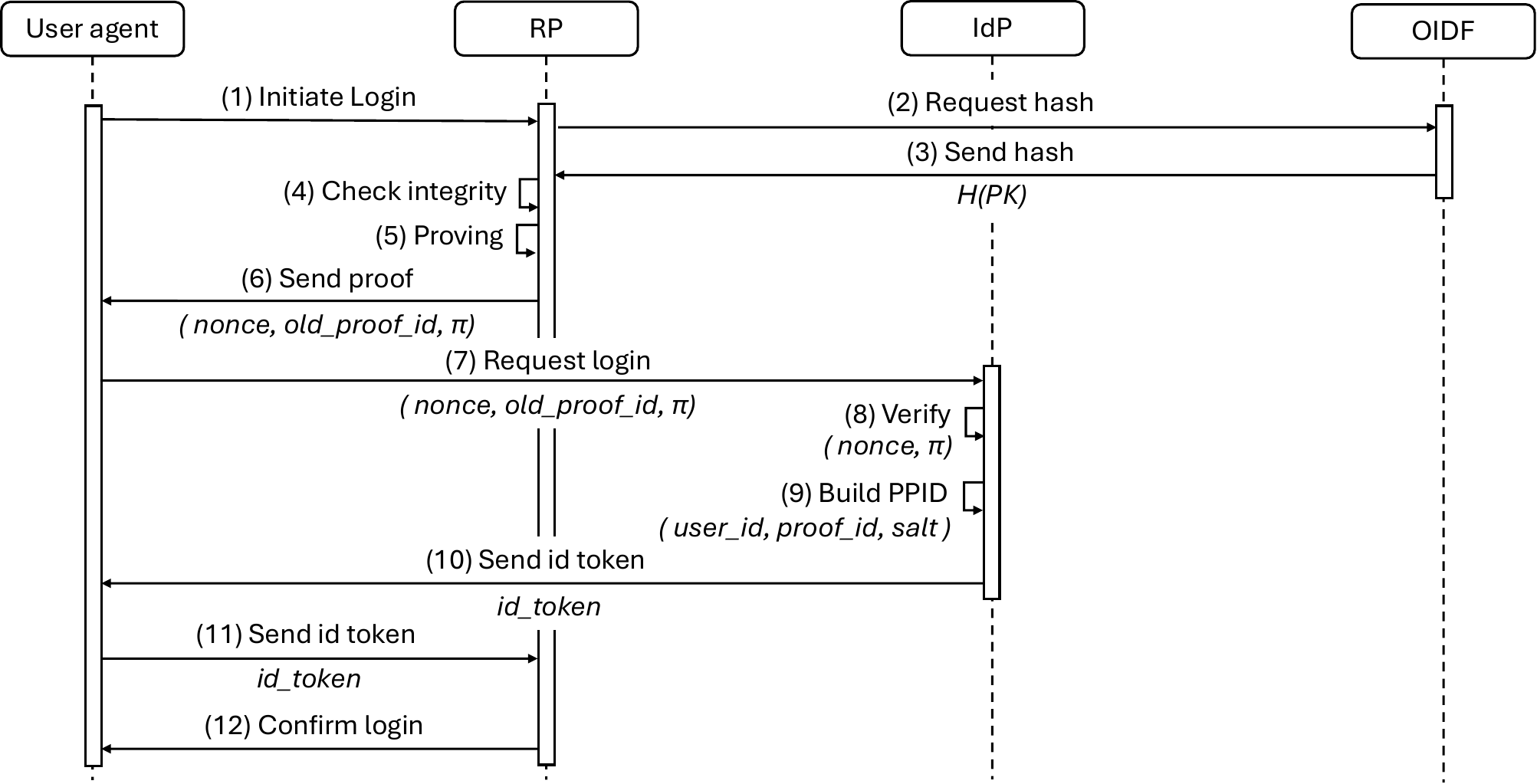}
  \caption{Sequential flow of information during user authentication phase}
  \label{fig:login_flow}
\end{figure}

\begin{theorem}\label{thm1}
Standardized circuits and published artifact hash values prevent a curious verifier from discovering the prover's identity.
\end{theorem}

\begin{proof}
    We assume that the circuit is not standardized. The verifier (IdP) will be controlling the circuit. This gives the verifier the power to manipulate the circuit and deceive the prover about the scope of the input. Also, if the hash value of zk-Artifacts is not published by OIDF, IdPs will be able to associate RPs with their proofs by assigning unique zk-Artifacts to each client. Let's say IdP requests $zk-Artifact_1$ from OIDF, extracts proving key $PK_1$, and provides it to $Client_1$ during the RP registration phase. It can further request $zk-Artifact_2$ and provide $PK_2$ to $Client_2$. When $Client_1$ sends $Proof_1$ during user authentication phase, the IdP can find out that only verification key $VK_1$ (also extracted from $zk-Artifacts_1$) can be used to verify $Proof_1$ successfully, hence revealing the mapped $Client_1$'s identity. By publishing the hash value of the artifacts, all clients will be able to check if they are using only the latest and a common proving key for proof generation. 
\end{proof}

\textbf{Client Membership Proof.}
Figure~\ref{fig:login_flow} shows the user authentication phase.
During the user authentication phase, the RP sends its membership proof to the IdP through the user agent.
The proof convinces the IdP that the RP possesses a valid \textit{client\_id} that has been signed by the IdP's private key. This achieves the completeness of ZKP. The proof does not reveal the \textit{client\_id} to the IdP, marking the zero-knowledge property as well. Lastly, it is infeasible for an external adversary without a valid \textit{client\_id} to generate a proof that can convince the IdP about membership, achieving the soundness property of ZKP.

\subsection{Proof-Based PPIDs}
We implement PPIDs for the privacy of end-users from colluding RPs. In ExPrESSO, IdPs map the user identifiers with \textit{proof\_id}, which is the hash value of membership proofs \textit{H(proof)}. The IdP computes a fixed PPID as \textit{H(user\_id} \texttt{||} \textit{proof\_id} \texttt{||} \textit{salt)} for a given user whenever a certain proof is provided. The PPID for a user remains the same through multiple login sessions for a specific RP. However, A different yet deterministic PPID is returned for the same user when a different proof (denoting a different RP) is received.

\begin{theorem}\label{thm2}
Proof-based PPIDs are unlinkable to unique users across multiple colluding-RPs in the random oracle model. 
\end{theorem}
\begin{proof}
    Given a specific user, $Client_1$ received $Id\_token_1$ from the IdP with $PPID_1$ as the \textit{user\_id}. $Client_2$ received $Id\_token_2$ from the same IdP and for the same user, but with $PPID_2$ as the \textit{user\_id}. When $Client_1$ and $Client_2$ collude and compare the PPIDs, they will be different random values from the hash function's domain, as the \textit{Proof\_id} for both RPs are different. Therefore, the PPIDs will not provide any information about the \textit{user\_id} of the user in the random oracle model. Furthermore, the \textit{user\_id} is an unpredictable random value, and thus no RP can efficiently compute a pre-image of the hash function. 
\end{proof}

All phases of ExPrESSO~\cite{ExPrESSO2025} including registration, and deregistration are shown in the appendix in figure~\ref{fig:all_flow} of section~\ref{sec:app_full_flow}.

\subsection{Security Analysis}
Now we formally analyze each of our five security and privacy aspects and discuss how ExPrESSO fulfils them.
We also describe how our model offers the expected characteristics.
\begin{enumerate}
\item \textbf{RP Unobservability: }
The RP unobservability property requires that, despite the RP being properly authenticated with the IdP, its identity remains hidden from the IdP.
ExPreSSO guarantees RP unobservability by relying on the zero-knowledge property of ZKPs. This ensures that no information about RPs is revealed to IdPs other than the validity of the RP's authentication.
Based on that, ExPrESSO considers two additional attack vectors:
First, a curious IdP might attempt to assign unique ZKP artifacts to each RP during registration or updates. 
The IdP could then resolve the 1-to-1 mapping between a presented ZKP and the RP's identity. In ExPrESSO, this is mitigated by Theorem~\ref{thm1}.
When artifacts need to be updated, the IdP requests new ZK artifacts from OIDF. 
RPs can further ensure they are using the legitimate artifacts by checking the integrity of ZK artifacts provided by the IdP against the public records maintained by OIDF. 
This prevents the IdP from distinguishing RPs based on the ZKP artifacts used for proof generation.
Second, direct communication from the RP to the IdP to submit the ZKP could expose the RP's IP address, allowing the IdP to identify the RP. 
In ExPrESSO, this is mitigated by relying on the OIDC implicit flow with the user agent forwarding the ZKP from an RP to the IdP as part of the authentication.

\item \textbf{User Unlinkability: }
In our work, the client identifies users through proof-based PPIDs instead of \textit{user\_id}. From Theorem~\ref{thm2}, we can assert that colluding clients cannot link the identity attributes of a user even if both clients hold an \textit{Id\_Token} of the same user. Note that this property is achieved without breaking unobservability by replacing \textit{Client\_Id} with \textit{Proof\_Id}. 

\item \textbf{Authentication Integrity. }
ExPreSSO preserves authentication integrity against malicious RPs by relying on the unforgeability of the EdDSA-based signature derived from \textit{client\_id} and the private key of IdP.
It is infeasible for malicious RPs to execute the proving without holding a valid authentication token issued by an IdP.
Based on the soundness property of the ZKP, no adversary can construct a proof to pass the RP authentication, since \textit{client\_id} is only known to the RP. 
Therefore, no adversary can successfully pass the authentication.

\item \textbf{RP Revocation: }
In our system, a RP can also deregister from an IdP's identity service. The de-registered RP stops receiving periodic updates for proving key $PK$ from the IdP. Simultaneously, as explained in Section~\ref{revocation}, when an RP deregisters, the IdP obtains a new set of zk-Artifacts. The IdP further extracts the new proving key $PK'$ and broadcasts it to all currently valid RPs. Without receiving $PK'$, the de-registered RP is unable to create a new valid proof and is unable to convince the IdP anymore about its membership, marking the achievement of RP revocation. 

\item \textbf{Trust Distribution: }
ExPrESSO distributes trust by reducing reliance on a single trusted entity.
As proven in Theorem~\ref{thm1}, a standardized circuit (derived from the standardized boilerplate program) can be the basis for preventing curious IdPs from discovering RP's identity. The standardized boilerplate program is formulated and maintained by the consensus of multiple entities in OIDF. As a result, the user and the RP realize a transparent and secure anchor for a high degree of trust. 
Furthermore, for RP registration, ExPrESSO employs the MPC setup ceremony to generate the cryptographic parameters required by the ZKP scheme.
This ensures that the integrity and privacy of ExPrESSO do not depend on a single party involved in the setup, as long as at least one honest participant contributes correctly.
The latter relieves the IdP from the responsibility of secret trapdoor parameter deletion in a trustworthy manner, further increasing the security of our zk-SNARK scheme adopted by our ExPrESSO.
\end{enumerate}

\section{Evaluation and Result}\label{evaluate}
We evaluated the prototype of ExPrESSO for its performance, including proof generation and proof verification. We record the average latency and the communication overhead for user authentication by running them 50 times. We tested the prototype on an Apple MacBook Pro machine equipped with an Apple Silicon M1 Pro processor, 16GB of Memory, and a 500GB SSD.  

\begin{table}[htbp]
\begin{threeparttable}
    \centering
    \caption{Comparison of execution time (in ms)}
    \begin{tabularx}{\textwidth}{X|XXXX}
    \toprule
    \makecell[l]{\textbf{Solution}} &
    \makecell[l]{\textbf{Proving}} & 
    \makecell[l]{\textbf{Verification}} &
    \makecell[l]{\textbf{OIDC Ops.}} &
    \makecell[l]{\textbf{User Auth.}}
    \\
    \midrule
    \makecell[l]{POIDC \cite{10.1145/3320269.3384724}} & -- & -- & 14000 & 13014 \\
    \makecell[l]{EL PASSO \cite{zhang2021passo}} & 11.56 & 13.41 & 800 & 824.97 \\
    \makecell[l]{AIF-ZKP \cite{kroschewski2023save}} & 25.63 & 18.98 & 9.86 & 54.47 \\
    \makecell[l]{ARPSSO \cite{10.1007/978-3-031-70890-9_14}} & 193.2 & 24.3 & -- & 217.5 \\
    \makecell[l]{OPPID \cite{kroschewski2025oppid}} & 15.26 & 17.33 & 1.47 & 34.06 \\
    \midrule
    ExPrESSO & 4338\tnote{*\textsuperscript{1}} & 237.30 & 1.8 & 239.2\tnote{*\textsuperscript{2}} \\
    \bottomrule
    \end{tabularx}
    \begin{tablenotes}
        \item[*\textsuperscript{1}] only executed once
        \item[*\textsuperscript{2}] includes \textit{Verification} and \textit{OIDC Ops.}
    \end{tablenotes}
    \label{tab:scheme_comparison}
\end{threeparttable}
\end{table}

ExPrESSO requires OIDF and the IdP to first initialize the prototype, which includes running at least one cycle of MPC-based trusted setup for a given IdP. The trusted setup phase is, however, an asynchronous process on which the user authentication does not depend entirely. Hence, we can conclude that the time duration and the communication overhead of the process do not affect or add to any of the other phase's runtime. The client registration phase is also independent and takes place separately prior to the user authentication phase. Hence, it does not affect the user authentication phase performance as well.


\textbf{Scaling Cost: }
Our R1CS arithmetic circuit has 94180 constraints. If OIDF decides to update the boilerplate program with additional operations (for example, by adding security features), it will change the number of constraints in the arithmetic circuit. The Groth16 prover’s runtime grows roughly linear with the number of constraints in the circuit (denoted by n)~\cite{salleras2020sans}. Proving is $O(n)$ – dominated by large multi-scalar multiplications and a few Fast Fourier Transforms (FFT) on polynomials. In practice, multi-scalar multiplication and FFT steps account for approximately 90\% of proving time~\cite{291319}. Groth16’s setup phase produces a proving key of size $O(n)$ and the prover must process all $n$ constraints, so both time and memory scale linearly in circuit size~\cite{salleras2020sans}. Updates to the high-level program do not affect the IdP's verification time, as Groth16 proofs are verified in constant time with respect to the circuit size.

\textbf{Operational Cost. }
We measured the time required to generate a proof by the RP and the proof verification time by the IdP. We further measured the total mean user authentication time required over 50 repetitions. Since the proof is generated for the same arithmetic circuit for every user authentication session, the RP does not need to re-generate it every time. Hence, the time required for proof generation has not been factored into the user authentication time in our case. Table~\ref{tab:scheme_comparison} presents a comparison of the performance of our system with the 5 most relevant works as discussed in section~\ref{related}.
Furthermore, the membership proof that the RP sends to IdP with the authentication request has a size of approximately 4KB. The size of the proof also remains the same even when scaling as discussed before. 

\section{Related Work}\label{related}
Since 2020, the research focus has shifted greatly toward OIDC-based privacy-preserving authentication systems.
We reviewed five OIDC-based relevant and comparable systems and marked the difference in their security properties in table~\ref{tab:comparision}.

POIDC~\cite{10.1145/3320269.3384724} hides the RP identity from IdP and prevents colluding RPs from linking users by generating one-time pseudonyms. However, it lacks a clear RP-registration phase, leaving uncertainty on how IdPs authenticate clients before authenticating end-users. This makes it vulnerable to malicious or unregistered clients obtaining tokens. EL PASSO~\cite{zhang2021passo} combines PS Signatures and NIZKs to provide efficient anonymous credentials with features like asynchronous login and selective disclosure. Its drawback is incompatibility with the OIDC flow since it omits RP registration, making it impossible for IdPs to restrict services to registered clients. AIF-ZKP~\cite{kroschewski2023save} enhances the OIDC implicit flow by adding ZKP-based RP authentication, protecting against phishing while hiding RP identity. Its limitation is reliance on short-lived, epoch-based credentials, which increase overhead, and it does not fully address colluding RP unlinkability. OPPID~\cite{kroschewski2025oppid} extends AIF-ZKP with pairwise pseudonymous identifiers to achieve unlinkability against colluding RPs. However, it requires end-users to perform cryptographic operations like opening commitments and signing requests, which demands technical expertise or extra tools — a significant usability limitation. ARPSSO~\cite{10.1007/978-3-031-70890-9_14} adapts privacy-preserving techniques to the OIDC code flow using anonymous credentials and in-browser forwarding modules to hide RP identity. Yet, its privacy guarantee depends on external tools like Tor for IP protection, which is impractical in real-world deployments.

\begin{table}[htbp]
    \centering
    \caption{Comparison of system properties}
    \scriptsize
    \begin{tabularx}{\columnwidth}{l|XXXXXX}
    \toprule
    \makecell[l]{\textbf{Solution /}\\\textbf{Reference}} &
    \makecell{\textbf{Client}\\\textbf{Unobser-}\\\textbf{vability}} &
    \makecell{\textbf{Authenti-}\\\textbf{cation}\\\textbf{Integrity}} &
    \makecell{\textbf{Distri-}\\\textbf{buted}\\\textbf{Trust}} &
    \makecell{\textbf{Unlink-}\\\textbf{ability}} &
    \makecell{\textbf{Client}\\\textbf{Revo-}\\\textbf{cation}} &
    \makecell{\textbf{No User}\\\textbf{Require-}\\\textbf{ments}}
    \\
    \midrule
    \makecell[l]{POIDC \cite{10.1145/3320269.3384724}} & \makecell{\fullcirc} & \makecell{\fullcirc} & \makecell{\emptycirc} & \makecell{\fullcirc} & \makecell{\emptycirc} & \makecell{\emptycirc} \\
    \makecell[l]{EL PASSO \cite{zhang2021passo}} & \makecell{\fullcirc} & \makecell{\fullcirc} & \makecell{\emptycirc} & \makecell{\fullcirc} & \makecell{\emptycirc} & \makecell{\emptycirc} \\
    \makecell[l]{AIF-ZKP \cite{kroschewski2023save}} & \makecell{\fullcirc} & \makecell{\fullcirc} & \makecell{\emptycirc} & \makecell{\emptycirc} & \makecell{\fullcirc} & \makecell{\emptycirc} \\
    \makecell[l]{ARPSSO \cite{10.1007/978-3-031-70890-9_14}} & \makecell{\fullcirc} & \makecell{\fullcirc} & \makecell{\emptycirc} & \makecell{\fullcirc} & \makecell{\emptycirc} & \makecell{\fullcirc} \\
    \makecell[l]{OPPID \cite{kroschewski2025oppid}} & \makecell{\fullcirc} & \makecell{\fullcirc} & \makecell{\emptycirc} & \makecell{\fullcirc} & \makecell{\fullcirc} & \makecell{\emptycirc} \\
    \midrule
    ExPrESSO & \makecell{\fullcirc} & \makecell{\fullcirc} & \makecell{\fullcirc} & \makecell{\fullcirc} & \makecell{\fullcirc} & \makecell{\fullcirc} \\
    \bottomrule
    \end{tabularx}
    \label{tab:comparision}
\end{table}

\section{Discussion and Conclusion}
Our evaluation shows that ExPrESSO is feasible and can be used in practice as long as a standardizing organisation, such as OIDF, identity providers, and service providers, mutually agree to adopt a privacy-preserving SSO system.
It is worth noticing that we assumed OIDF to maintain a pool of zk-SNARK artifacts for fast responses to IdP. If OIDF maintains a pool of $n$ artifacts and $k$ IdPs exist that use ExPrESSO, in the worst-case scenario,  all of $k$ IdPs might request new artifacts multiple times in a short period of time. If the total number of requests exceeds $n$, then the  $n+1^{st}$ and later IdPs will have to wait for OIDF to perform an MPC-based trusted setup before they can be issued new artifacts. Hence, OIDF should maintain a pool of size $n$ much larger than $k$. Furthermore, our work includes the assumption that OIDF performs the \textit{Powers of Tau} ceremony by importing parameters from an externally operated, decentralized and verifiable source. OIDF sequentially performs the second step through MPC within the participating entities. Both steps require generating parameters evaluated on a random beacon with high entropy and sourced from publicly available data~\cite{10.1007/978-3-031-54776-8_5}. Examples of such randomized data can be the stock market value of a specific day or the value of a block at a particular height in a blockchain.

We have presented an OIDC-based authentication system that preserves the privacy of end-users during SSO-based login sessions of users. We assume OIDF to play a bigger role as a decentralized trust anchor in a zero-trust environment. ExPrESSO leverages properties of Groth's zkSNARK to hide the identity of clients from IdPs. We achieved all of our design goals and target characteristics without involving significant disk space and communication overhead. Users can have privacy-preserving SSO authentication without requiring any additional tools or technical knowledge and with only a small amount of latency. We observe from table~\ref{tab:scheme_comparison} that the user-login time is comparable, if not faster than most other privacy-preserving SSO systems. Additionally, ExPrESSO offers security properties unmatched by any other work. Hence, we can infer that ExPrESSO can be deployed for better privacy in the real world for a little cost of user authentication time (by milliseconds). In the future, we plan to extend this work to enable privacy for users in OIDC code flow as well. 

\newpage

\bibliography{template}

\newpage

\begin{subappendices}

\section{ExPRESSO registration, login, and deregistration flow}
\label{sec:app_full_flow}

\begin{figure}[htbp] 
  \centering
  \includegraphics[width=\textwidth]{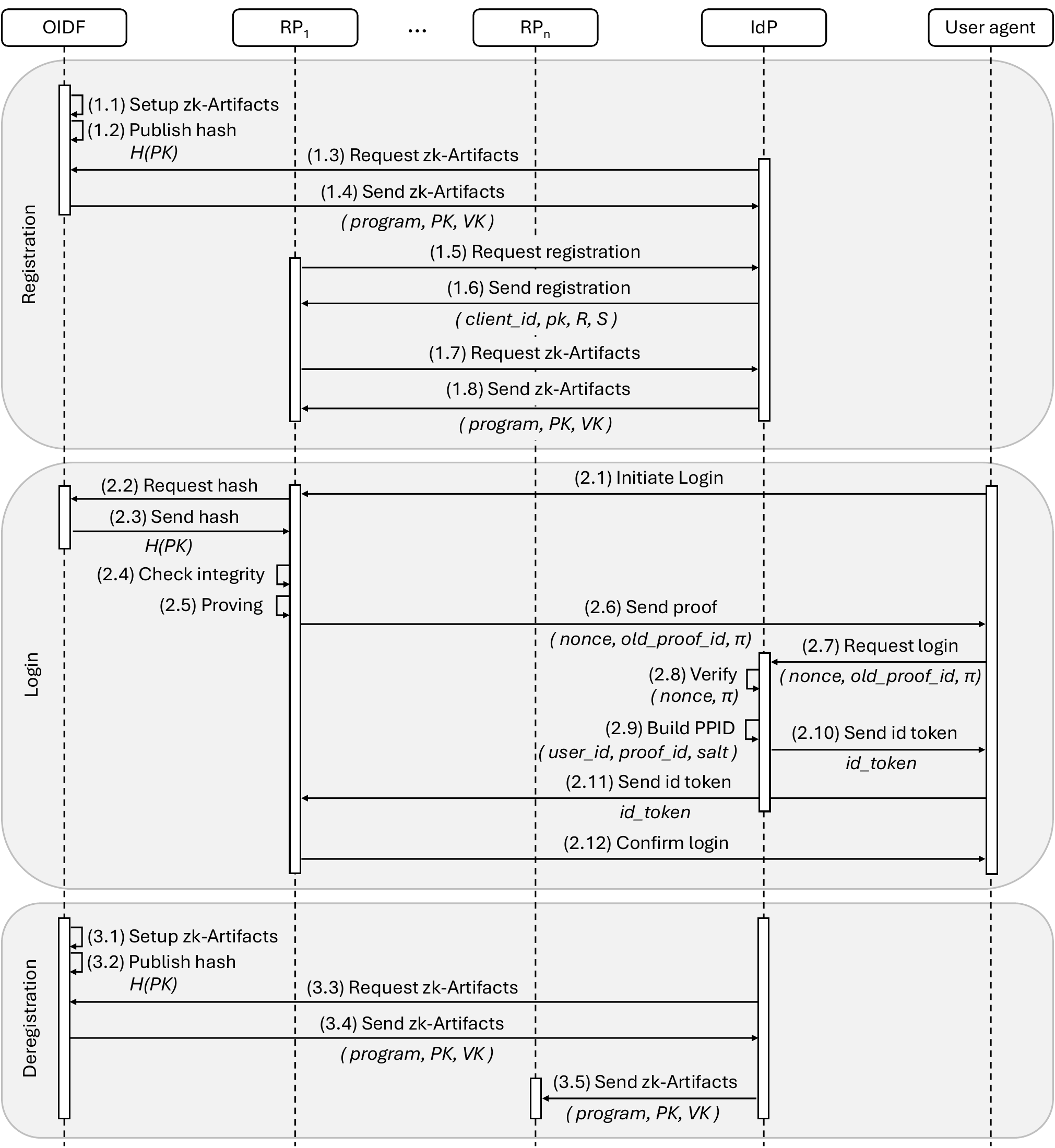}
  \caption{Sequential flow of information during the registration, login, and deregistration phase}
  \label{fig:all_flow}
\end{figure}




\end{subappendices}

\end{document}